\documentclass[pre,aps,twocolumn,showpacs,superscriptaddress,nofootinbib]{revtex4-1}
\usepackage{graphicx}
\usepackage{amssymb}
\usepackage{amsmath}

\usepackage{epstopdf}
\usepackage{epsfig}
\usepackage{pictex}

\newtheorem{theorem}{Theorem}

\newtheorem{lemma}[theorem]{Lemma}

\newcommand{\qed}{\hfill $\blacksquare$}

\newcommand{\ncr}[2]{\left(
\begin{array}{c}
#1 \\
#2 \\
\end{array} \right)}

\begin{document}

\title{Encoding Universal Computation in the Ground States of Ising Lattices}

\author{Mile Gu}
\affiliation{Center for Quantum Technologies, National University of Singapore, Singapore.}
\author{\'Alvaro Perales}
\affiliation{Computer Engineering Department, Universidad de Alcal\'a, Madrid 28871, Spain.}

\date{\today}

%%%%%%%%%%%%%%%%%%%%%%%%%%%%%%%%%%%%%%%%%%%%%%%%%%%%%%%%%%%%%%%%%%%%%%%%%%%%%%%%%%%%%%%%%%%%%%%%%%%%%%%%%%%%%%%%%%%%%%%%%%%%%%%%%%
%%%%%%%%%%%%%%%%%%%%%%%%%%%%%%%%%%%%%%%%%%%%%%%%%%%%%%%%%%%%%%%%%%%%%%%%%%%%%%%%%%%%%%%%%%%%%%%%%%%%%%%%%%%%%%%%%%%%%%%%%%%%%%%%%%

\begin{abstract}

We characterize the set of ground states that can be synthesized by classical $2$-body Ising Hamiltonians. We then construct simple Ising planar blocks that simulates efficiently a universal set of logic gates and connections, and hence any Boolean function.  We therefore provide a new method of encoding universal computation in the ground states of Ising lattices, and a simpler alternative demonstration of the known fact that finding the ground state of a finite Ising spin glass model is NP complete. We relate this with our previous result about emergent properties in infinite lattices.

\end{abstract}

%\pacs{03.67.Dd, 42.50.Dv, 89.70.+c}

\maketitle

%%%%%%%%%%%%%%%%%%%%%%%%%%%%%%%%%%%%%%%%%%%%%%%%%%%%%%%%%%%%%%%%%%%%%%%%%%%%%%%%%%%%%%%%%%%%%%%%%%%%%%%%%%%%%%%%%%%%%%%%%%%%%%%%%%
%%%%%%%%%%%%%%%%%%%%%%%%%%%%%%%%%%%%%%%%%%%%%%%%%%%%%%%%%%%%%%%%%%%%%%%%%%%%%%%%%%%%%%%%%%%%%%%%%%%%%%%%%%%%%%%%%%%%%%%%%%%%%%%%%%
\section{Introduction}

The Physical Church-Turing thesis \cite{Penrose89a} provides a deep connection between the science of computation and the physical universe. It posits that the dynamics of any known physical system can be simulated by a Turing machine \cite{Turing36a}, a theoretical device that consists of a finite state machine together with an infinite tape.  Upon reflection, this is a remarkable result, widely believed to be correct. An arbitrary physical system is governed by a vast variety of different forces, from Coulomb interactions to gravity, and there is no reason, a priori, to suspect that all of these effects can be replicated on one particular machine. This presents the idea of universality: a physical system is universal if its dynamics can be used to simulate any other physical system.

The prevalence of universality in commonly studied systems is not only a theoretical curiosity, but also has consequences of practical significance. Recent results in computer science restrict our ability to predict the behavior of such systems. Observations of universal systems led to the Strong Church-Turing thesis \cite{Bernstein97a}, which postulates that a Turing Machine together with a source of randomness is computationally as powerful as any other existing universal system. Formally speaking, we say that a task lies in P, or is tractable, if the task can be performed efficiently by a Turing Machine, i.e., the time required to perform it scales as a polynomial of the size of the input \cite{Papadimitriou94a}. This thesis then postulates that any problem which lies outside P cannot be solved with resources that scale polynomially with respect to the size of the problem, regardless of the method of computation used. While the existence of Shor's algorithm in quantum mechanics may provide an exception to this thesis \cite{Shor97a}, it applies to all current classical models of computation.

This leads to deep insights into any universal system that simulates a Turing machine efficiently. Suppose such a system simulates a Turing machine operating on an intractable problem as input. If one could efficiently compute every physical property of this system, then one can use it to solve the encoded problem and therefore violate the Strong Church-Turing thesis. Thus, such universal systems must necessarily exhibit properties which no classical algorithm can efficiently compute.

Many other universal systems have been proposed, for example, logic circuits \cite{Nielsen00a}, the Game of Life \cite{Conway82}, Rule 110 \cite{cook04a}, and measurement based quantum computation \cite{Raussendorf01a}.
In addition to these abstract mathematical constructs, many surprisingly simple physical systems capable of universal computation have also been discovered. These include billiard balls \cite{Fred82}, simple dynamical systems \cite{Moore90a} and the dynamics of $3$-dimensional majority voting cellular automata \cite{Moore97a}.

This motivates an interesting question: how simple can a physical system be to still exhibit universality and thus complex behaviour? In particular, we explore what limits can be placed on a class of Hamiltonians such that the evaluation of their ground states still requires the capacity to perform universal computation. We relate this to the ground state decision problem: given a Hamiltonian $H$ and some number $E$, does there exist a state with energy at most $E$?

Interestingly, the ground state decision problem is difficult to solve even for the simple Ising lattice, which is a widely used model to describe collective behaviour in diverse systems, as magnetism \cite{Chandler87a}, lattice gases \cite{Lee87a}, neural activity \cite{Rojas} and even protein folding \cite{Bryngelson91}. While an efficient solution is known in the case of one dimension, F. Barahona showed in 1982 that the computational task is generally NP-complete in higher dimensions \cite{barahona82a}, whenever some of the bonds are antiferromagnetic. Here, NP denotes the class of non-deterministic polynomial time problems; an abstract class of problems whose solutions can be verified, but not necessarily found, in polynomial time. Indeed, this connection has even allowed the engineering of spin lattice Hamiltonians whose ground states help model and study NP-complete problems \cite{Monasson99a}.

The complexity of the ground state decision problem suggested that such ground states could also embed universal computation. Indeed, this was first proven with the adiabatical model of quantum computation, where a simple Hamiltonian with known ground state is adiabatically evolved to the complex Hamiltonian whose ground state encodes the solution to the computational problem \cite{Farhi00a, Farhi01a}. To further simplify the models and make them more suitable to be recreated in real experiments, it has been proven that it is enough to consider just $2$-body interactions in the Hamiltonian to obtain the capability of universal computation \cite{Oliveira08a, Kempe06a, Biamonte08a, Biamonte08b}.\\

In this paper we extend those studies in the classical case and derive a general result on what ground state sets can be synthesized by a $m$-body Hamiltonian on a system of $n$ spins. Using the circuit model of computation, we construct simple designer circuit blocks that can be combined to encode a universal computer in the ground state of $2$-body Ising Hamiltonians, in such a way that there is a map between any given logic circuit to the ground states of some Hamiltonian. This encoding, together with the strong Church-Turing thesis, provides immediate implications on the computational complexity of evaluating such ground states. Furthermore, this allows us to provide a simple alternative proof of Barahona's result that the ground state decision problem is NP-complete \cite{barahona82a}.

We explore the connection of this result with the infinite lattice case we studied in a previous work \cite{Gu09a}. We showed that there are undecidable properties in the infinite Ising model that give rise to emergent properties in the physical Ising lattice. Besides, the circuit blocks presented here simplify the technical parts of that work.\\

%%%%%%%%%%%%%%%%%%%%%%%%%%%%%%%%%%%%%%%%%%%%%%%%%%%%%%%%%%%%%%%%%%%%%%%%%%%%%%%%%%%%%%%%%%%%%%%
%%%% STRUCTURE

This paper is organized as follows. Section \ref{notation} introduces the required background and notation. Section \ref{designer} introduced the ground synthesis problem, whilst Section \ref{universality} gives an alternate proof of the universality of Ising ground states. Section \ref{complexity} explores the consequences in complexity of the computational difficulty of the ground state problem and the the relation with the infinite case and emergence. Section \ref{conclusion} presents the main conclusions.
%%%%%%%%%%%%%%%%%%%%%%%%%%%%%%%%%%%%%%%%%%%%%%%%%%%%%%%%%%%%%%%%%%%%%%%%%%%%%%%%%%%%%%%%%%%%%%%%%

%%%%%%%%%%%%%%%%%%%%%%%%%%%%%%%%%%%%%%%%%%%%%%%%%%%%%%%%%%%%%%%%%%%%%%%%%%%%%%%%%%%%%%%%%%%%%%%%%%%%%%%%%%%%%%%%%%%%%%%%%%%%%%%%%%
%%%%%%%%%%%%%%%%%%%%%%%%%%%%%%%%%%%%%%%%%%%%%%%%%%%%%%%%%%%%%%%%%%%%%%%%%%%%%%%%%%%%%%%%%%%%%%%%%%%%%%%%%%%%%%%%%%%%%%%%%%%%%%%%%%
\section{Background and Notation} \label{notation}

To explore how ground states can embed universal computation, we first address a related practical problem of ground state synthesis, i.e., given a set of desired states, is it possible to engineer a Hamiltonian whose set of ground states correspond to those in the desired set? In particular, when reality dictates certain limits on the interactions available, what are the corresponding restrictions on the possible ground states that can be achieved? For example, denote the state of each spin by either 0 or 1, is it possible to find a Hamiltonian with ground states given by $\{000, 011, 101, 110\}$? If so, is it possible to engineer this Hamiltonian from Ising interactions? The solution to the above question gives us the tools to engineer a set of states that are capable of encoding a universal circuit.\\

Let us first define the nomenclature used in this paper. Denote the state of each spin by either $0$ or $1$. A system of $n$ spins is described by a binary number $\mathbf{b} = b_1\ldots b_n \in \mathbb{Z}_2^n$, where $b_i \in \{0,1\}$ denotes the state of the $i^{th}$ spin. Given a state $\mathbf{b}$, we make the following definitions:

\begin{itemize}
 \item \textbf{Weight}: $|\mathbf{b}|$ is the number of $1$s in $\mathbf{b}$.
 \item  \textbf{1-Sites} Ones(\textbf{b}) is the set of indices whose corresponding spins are $1$. $\mathrm{Ones}(\mathbf{b}) := \{i: \, b_i = 1\}$.
 \item \textbf{Descendant} $\mathbf{a}$ is a descendent of $\mathbf{b}$ iff $\mathrm{Ones}(\mathbf{a}) \subseteq \mathrm{Ones}(\mathbf{b})$, i.e., the $1$-sites of $\mathbf{a}$ are subsets of $1$-sites of $\mathbf{b}$. We write this as a partial order, $\mathbf{a} \preceq \mathbf{b}$. $\mathrm{Dsc}(\mathbf{b})$ defines the set of all descendants of $\mathbf{b}$, and $\mathrm{Dsc}(\mathbf{b},k) := \{\mathbf{a}: \mathbf{a} \preceq \mathbf{b}, |\mathbf{a}| = k\}$ are all descendants of $\mathbf{b}$ with weight $k$.
\end{itemize}

A Hamiltonian on this system is defined by a function $H: \mathbb{Z}_2^n \rightarrow \mathbb{R}$ that maps each state of the system to a corresponding energy. A general Hamiltonian is of the form:
\begin{equation}\label{eqn:bob}
H(b_1,\ldots,b_n) =
\sum_{\mathbf{\mathbf{a}} \in \mathbb{Z}_2^n}c_\mathbf{\mathbf{a}}b_1^{a_1}b_2^{a_2}\ldots
b_n^{a_n},
\end{equation}

\noindent where $\mathbf{a} = a_1a_2\ldots a_n \in \mathbb{Z}_2^n$, $a_i \in \{0,1\}$, $c_{\mathbf{a}}$ are arbitrary constants, and the summation is taken over all binary strings of length $n$. Since we can always choose a labeling of the spin states such that one of the ground state corresponds to $\mathbf{0}$, we assert that $\mathbf{0}$ is a ground state of $H$ (i.e. $H(\mathbf{0}) = 0$) without loss of generality.

A Hamiltonian $H$ is $m$-body if it does not contain interactions involving $m+1$ spins or greater, i.e: $c_{\mathbf{a}} = 0$ $\forall \mathbf{a}$ such that $|\mathbf{a}| > m$. The general Ising model with an external magnetic field is a $2$-body Hamiltonian of the form \cite{Lee87a}:
\begin{equation}
 H = \sum c_{jk} b_jb_k  + \sum M_j b_j,
\end{equation}

\noindent where $c_{jk}$ are the interaction energies between spins $j$ and $k$, and $M_j$ describes the external field at site $j$.

Interaction graphs provide a convenient tool to visualize Ising Hamiltonians. Given a system of $n$ spins, we associate with it a graph of $n$ vertices where each spin corresponds to a single vertex. We draw an edge between two vertices $v_i$ and $v_j$ if the interaction energy between them, $c_{jk}$ is non-zero. A square Ising model of size $N$ is described by an interaction graph with vertices $v_{j,k}$ where $j,k = 1,\ldots,N$, with edge set $E = \{(v_{j,k},v_{j+1,k}), (v_{j,k},v_{j,k+1})\}$ with $j,k = 1,\ldots,N+1$.

The main idea of our approach is as follows. To embed a binary function on two bits $b_{out} = f(b_1,b_2)$, we construct a Hamiltonian $H_f$ on $b_1$,$b_2$,$b_{out}$ with the ground state set
\begin{equation}
 \mathcal{G}_{f} = \{00f(00),01f(01),10f(10),11f(11)\}.
\end{equation}
We see that each element of $\mathcal{G}_{f}$ satisfy $b_{out} = f(b_1,b_2)$. We define the spins in state $b_1$ and $b_2$ as input spins, and the bit in state $b_{out}$ as the output spin. We say that the ground state $\mathcal{G}_{f}$ \emph{encodes} $f$.

We can then evaluate the action of $f$ on particular input, i.e., $f(x,y)$, by introducing the external biases on the input spins that breaks the degeneracy of $H_f$ such that the state $x,yf(x,y)$ has lower energy than the other elements of $\mathcal{G}$. For example, the Hamiltonian $H_{f(00)} = H_f + b_1 + b_2$ would have the unique ground state $\{00f(00)\}$. Therefore, cooling such a system to ground state would allow us evaluate $f(0,0)$, and the computational task of solving for a ground state of this system is at least as hard as evaluate $f(0,0)$.

%%%%%%%%%%%%%%%%%%%%%%%%%%%%%%%%%%%%%%%%%%%%%%%%%%%%%%%%%%%%%%%%%%%%%%%%%%%%%%%%%%%%%%%%%%%%%%%%%%%%%%%%%%%%%%%%%%%%%%%%%%%%%%%%%%
%%%%%%%%%%%%%%%%%%%%%%%%%%%%%%%%%%%%%%%%%%%%%%%%%%%%%%%%%%%%%%%%%%%%%%%%%%%%%%%%%%%%%%%%%%%%%%%%%%%%%%%%%%%%%%%%%%%%%%%%%%%%%%%%%%
\section{Ground State Synthesis}  \label{designer}

This motivates the problem of ground state synthesis, i.e., given a set of desired states, is it possible to engineer an $m$-body Hamiltonian with a coinciding set of ground states, and if so, how? The answer of this question can be directly applied to \textit{designer ground states}, a set of ground states $G_f$ specifically designed to encode a desired binary function $f$. Should we be able to construct $m$-body Hamiltonians for arbitrary $f$, we can establish the universality of the Ising model.

We can represent $H(\mathbf{b})$ and  $c_{\mathbf{b}}$ as vectors in $\mathbb{R}^{2n}$, where their components are indexed by all possible values of $\mathbf{b} \in \{0,1\}^{n}$. Eq.~(\ref{eqn:bob}) implies that $H(\mathbf{b}) = L c_{\mathbf{b}}$, where $L$ is some invertible linear map. Thus, the restriction of $H$ to $m$-body interactions leads to a set of linear equations that constrain $H(\mathbf{b})$. More precisely, $H$ is an $m$-body Hamiltonian iff for each $H(\mathbf{b})$ with $|\mathbf{b}| = k > m$,
\begin{align}\label{eqn:hrelmain}
H(\mathbf{b}) &= \sum_{p=1}^m a_p \left[ \sum_{\mathbf{d} \in
\mathrm{Dsc}(\textbf{b},p)} H(\mathbf{d}) \right],
\end{align}

\noindent where $a_p$ is given by the recurrence relation (see appendix):
\begin{align}
a_p =  \left \{\begin{array}{ll} 1 & p = m\\
1 - \sum_{j=1}^{m-p}a_{p+j}\ncr{|\mathbf{b}|-p}{j} & 1 \leq p < m\\
\end{array} \right. \label{eqn:arel}
\end{align}

This leads immediately to constraints on the ground state set $\mathcal{G}$ if it can be $m$-synthesized:

%%%%%%%%%%%%%%%%%%%%%%%%%%%%%%%%%%%%%%%%%%%%%%%%%%%%%%%%%%%%%%%%%%%%%%%%%%%%%%%%%%%%%%%%%%%%%%%%%%%%%%%%%%%%%%%%%%%%%%%%%%%%%%%%%%
\begin{theorem}\label{cornm}
Suppose $H$ is an $m$-body Hamiltonian on a system of $n$ spins. For each $\mathbf{b}$ with $|\mathbf{b}| = k > m$, define the sets $\mathcal{A} = \{\mathbf{b}\} \cup \mathrm{Dsc}(\mathbf{b},m-1) \cup \mathrm{Dsc}(\mathbf{b},m-3) \cup \ldots $ and  $\mathcal{B} = \mathrm{Dsc}(\mathbf{b},m) \cup \mathrm{Dsc}(\mathbf{b},m-2) \cup \ldots$ Then the ground state set $\mathcal{G}$ of $H$ must satisfy:
\begin{align}\mathcal{A} \subset \mathcal{G} \Leftrightarrow \mathcal{B} \subset \mathcal{G}
\end{align}
for every $\mathbf{b}$ with $k > m$.
\end{theorem}

\begin{proof}Observe that $a_p$ alternates signs for each value of $p$ in Eq.~(\ref{eqn:arel}), thus we can write Eq.~$(\ref{eqn:hrelmain})$ in the form $\sum_{\mathbf{b} \in \mathcal{A}} c_{\mathbf{b}} H(\mathbf{b}) = \sum_{\mathbf{b} \in \mathcal{B}} c_{\mathbf{b}} H(\mathbf{b})$. If $\mathcal{A} \subset \mathcal{G}$, then the left hand of this equation is $0$. Since $H(\mathbf{b}) \geq 0$ by assumption, it follows that the right hand side must also be $0$, and vice versa. \qed
\end{proof}\\

This theorem immediately implies that restrictions to $m$-body Hamiltonians, for any $m$, will also restrict the sets of ground states that we can synthesize. In particular, an $m$-body can only implement $m$-wise correlations. Consider for example the case of an $n$-body system, then any ground state set $\mathcal{G}$ that does not satisfy
\begin{equation}\label{eqn:wholebody}
\{\mathbf{b}: \mathrm{wt}(\mathbf{b}) \textrm{ odd} \} \subseteq
\mathcal{G} \Leftrightarrow \{\mathbf{b}: \mathrm{wt}(\mathrm{b})
\textrm{ even}\} \subseteq \mathcal{G}
\end{equation}
can only be synthesized by a Hamiltonian with all $n$ bodies interacting together. One observes that the ground state set corresponding to the parity function on a binary string (i.e: $f(\mathbf{b}) = |\mathbf{b}| \mod 2$) violates the above condition, and hence cannot be simulated by any $2$-body Hamiltonian. Thus, we cannot simulate all binary functions directly.

The above problem can be circumvented by introducing ancillae, additional bits within the Ising lattice that are not designated as either input or output bits. For example, consider simulation of the NAND gate, defined by $\mathrm{NAND}(b_1,b_2) = (b_1 \otimes b_2) \oplus 1$, where all arithmetic is done modulo $2$. Directly, a Hamiltonian $H_{NAND}$ with ground state set $\mathcal{G}_{\mathrm{NAND}} = \{001, 011, 101, 110\}$ simulates NAND. However, NAND can also be simulated any Hamiltonian on $k + 3$ spins, with a ground state set of the form  $\mathcal{G} = \{00\mathbf{s}_{00}1, 01\mathbf{s}_{01}1, 10\mathbf{s}_{10}1, 11\mathbf{s}_{11}0\}$, where each $\mathbf{s}_{ij}$ denote binary strings of length $k$.

Now consider binary functions $f$, $g$, $h$, simulated by Hamiltonians $H_f,H_g,H_h$, with outputs $b_f$, $b_g$ and $b_h$. The functional composition $f(g(b_1,b_2),h(b_3,b_4))$ on the four input bits $b_i$ where $i = 1,\ldots 4$, can be simulated by the Hamiltonian $H_g(b_1,b_2,b_g) + H_h(b_3,b_4,b_h) + H_f(b_g,b_h,b_{out})$, where $b_g$ and $b_h$ are introduced as ancillae.

%%%%%%%%%%%%%%%%%%%%%%%%%%%%%%%%%%%%%%%%%%%%%%%%%%%%%%%%%%%%%%%%%%%%%%%%%%%%%%%%%%%%%%%%%%%%%%%%%%%%%%%%%%%%%%%%%%%%%%%%%%%%%%%%%%
%%%%%%%%%%%%%%%%%%%%%%%%%%%%%%%%%%%%%%%%%%%%%%%%%%%%%%%%%%%%%%%%%%%%%%%%%%%%%%%%%%%%%%%%%%%%%%%%%%%%%%%%%%%%%%%%%%%%%%%%%%%%%%%%
\section{Universality of Ising Ground States} \label{universality}

An arbitrary Boolean circuit that takes $n$ input bits and maps them to $m$ output bits can be decomposed a basic logic circuit composed of the following components: \textit{Wires} that takes a spin as input, and copies its state to a neighboring spin; and NAND Gate that can generate all Boolean functions. These require the synthesis of the following ground state sets $\mathcal{G}_{WIRE} = \{00, 11\}$ and $\mathcal{G}_{NAND} = \{001, 011, 101, 110\}$ (In standard literature, the FANOUT gate that copies an input bit onto two outputs spins is also normally required. However, this operation can be decomposed in spin systems into two wires that connect to the same input spin.)

We the convert this to a \emph{planar circuit}, that is, one in which no wires may intersect. This requires the replacement of each section where a wires intersects with a SWAP gate, $\mathrm{SWAP}(b_1,b_2) = (b_2,b_1)$. We observe that this operation can be decomposed into a network of three XOR gates i.e: $\mathrm{SWAP}(b_1,b_2) = \mathrm{XOR}_1(\mathrm{XOR}_2(\mathrm{XOR}_1(b_1,b_2)))$, where $\mathrm{XOR}_1(b_1,b_2) = (b_1 \oplus b_2, b_2)$ and $\mathrm{XOR}_2(b_1,b_2) = (b_1 , b_1 \otimes b_2)$. We call this the planar circuit representation of $f$.

Therefore, we can construct a square Ising Hamiltonian that synthesizes $f$ provided there exists square Ising Hamiltonians that implement each of basic aforementioned components, i.e., (1) wires, (2) NAND gates and (3) XOR gates. To see that each of these can be simulated by a $2$-body Hamiltonian, we prove the following lemma:

%%%%%%%%%%%%%%%%%%%%%%%%%%%%%%%%%%%%%%%%%%%%%%%%%%%%%%%%%%%%%%%%%%%%%%%%%%%%%%%%%%%%%%%%%%%%%%%%%%%%%%%%%%%%%%%%%%%%%%%%%%%%%%%%
\begin{lemma}\label{thrm:32problem}Given a set of states $\mathcal{G}$ on a system of three spins with $000 \in \mathcal{G}$, there exists a $2$-body Hamiltonian that synthesizes $\mathcal{G}$ if and only if
\begin{equation}\label{eqn:32rel}
\{111\} \cup \mathrm{Desc}(111,1) \subseteq \mathcal{G} \Leftrightarrow  \mathrm{Desc}(111,2) \subseteq \mathcal{G}.
\end{equation}
\end{lemma}
\begin{proof}The forward direction is special case of Eq.~(\ref{eqn:wholebody}) for $n = 3$. To observe the converse, assume Eq.~(\ref{eqn:32rel}) is true. Eq.~(\ref{eqn:hrelmain}) implies that $H$ is a $2$-body Hamiltonian iff $H$ satisfies:
\begin{equation}\label{eqn:H32rel}
\sum_{\mathbf{b} \in \mathcal{A}}H(\mathbf{b}) = \sum_{\mathbf{b} \in \mathcal{B}} H(\mathbf{b}),
\end{equation}
where $\mathcal{A}  = \{111\} \cup \mathrm{Desc}(111,1)$ and $\mathcal{B} = \mathrm{Desc}(111,2)$. To see that Eq.~($\ref{eqn:H32rel}$) is true, observe that if $\mathcal{A}, \mathcal{B} \subseteq \mathcal{G}$ then Eq.~(\ref{eqn:H32rel}) is satisfied trivially. Otherwise, construct the Ising Hamiltonian that has assignments
\begin{equation}
 H(\mathbf{b}) = \frac{1}{|\mathcal{A}/\mathcal{G}|}, \  \  \  \  \  \  H(\mathbf{d}) = \frac{1}{|\mathcal{B}/\mathcal{G}|},
\end{equation}
for all $\mathbf{b} \in \mathcal{A}/\mathcal{G}$, $\mathbf{d} \in \mathcal{B}/\mathcal{G}$. Here $|\mathcal{A}/\mathcal{G}|$ is the number of elements that lie in $\mathcal{A}$ but outside $\mathcal{G}$. \qed
\end{proof}\\

The above lemma gives us a method to construct all the elements of a universal circuit from $2$-body nearest neighbor Hamiltonians. Wires can be simulated through $H_{I} = b_1 + b_2 - 2b_1b_2$. Lemma \ref{thrm:32problem} implies that the NAND can be simulated directly (to see NAND can be simulated, relabel the third bit). XOR cannot be implemented by the ground state of a $2$-body Hamiltonian on three spins. However, the Hamiltonian on four spins
\begin{align}\nonumber
H_{XOR}(b_1,b_2,b_A,b_o) &= (4 b_A - 3)(b_1 + b_2 + b_o) - 4 b_A\\
                    & + 2(b_1b_2 + b_2b_o + b_1b_o) + 4
\end{align}
with ground states $\{0010,0101,1001,1100\}$ simulates XOR using $b_A$ as an ancilla. Thus, all the above gates can be simulated by two-body Hamiltonians. Since these Hamiltonians also involve at most four spins, their interaction graphs must also be planar with vertices of degree at most three. Thus they can all be embedded in a square Ising Lattice with additional ancillae (See Fig. \ref{fig:ising_mapping}), and hence so can $f$.

Finally, we observe that each gate can be simulated by a Hamiltonian on at most $k$ spins, where $k$ is a fixed number. Thus, the number of spins used to simulate $f$ is at most some polynomial for the number of logic gates used to construct $f$. Therefore, the square Ising model can simulate an arbitrary circuit efficiently.

%%%%%%%%%%%%%%%%%%%%%%%%%%%%%%%%%%%%%%%%%%%%%%%%%%%%%%%%%%%%%%%%%%%%%%%%%%%%%%%%%%%%%%%%%%%%%%%%%%%%%%%%%%%%%%%%%%%%%%%%%%%%%%%%
\begin{figure}[htp]
\centering
\includegraphics[width=0.45\textwidth]{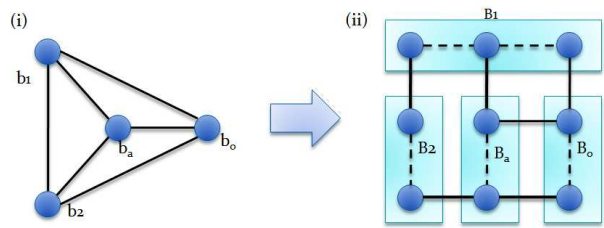}
\caption{The Hamiltonian that synthesizes the XOR Gate, $H_{XOR}$, with its corresponding interaction graph (i) can be embedded into a $3\times 3$ square Ising Lattice. (ii)  Each original spin $b_i$ is mapped to a set of spins $B_i$ which are linked by $H_{WIRE}$ interactions. At the ground state, all spins in each set $B_i$ are of the same state, and hence behave as if they are a single bit.}\label{fig:ising_mapping}
\end{figure}
%%%%%%%%%%%%%%%%%%%%%%%%%%%%%%%%%%%%%%%%%%%%%%%%%%%%%%%%%%%%%%%%%%%%%%%%%%%%%%%%%%%%%%%%%%%%%%%%%%%%%%%%%%%%%%%%%%%%%%%%%%%%%%%%%%

%%%%%%%%%%%%%%%%%%%%%%%%%%%%%%%%%%%%%%%%%%%%%%%%%%%%%%%%%%%%%%%%%%%%%%%%%%%%%%%%%%%%%%%%%%%%%%%%%%%%%%%%%%%%%%%%%%%%%%%%%%%%%%%%
\begin{theorem}\label{thm:circuit_gstate} Consider an arbitrary binary function $f$. There always exists a square Ising Hamiltonian $H$ whose ground states encode $f$.
\end{theorem}

The above theorem allows us to encode any logic circuit, and thus computational task, into the ground state of an Ising Hamiltonian. Not only is it remarkable that the ground state of such simple lattices are capable of simulating all physical processes, but this fact also allows us to apply the many results of computational complexity directly onto the task for computing ground states for an Ising Hamiltonian.

%%%%%%%%%%%%%%%%%%%%%%%%%%%%%%%%%%%%%%%%%%%%%%%%%%%%%%%%%%%%%%%%%%%%%%%%%%%%%%%%%%%%%%%%%%%%%%%%%%%%%%%%%%%%%%%%%%%%%%%%%%%%%%%%%%
%%%%%%%%%%%%%%%%%%%%%%%%%%%%%%%%%%%%%%%%%%%%%%%%%%%%%%%%%%%%%%%%%%%%%%%%%%%%%%%%%%%%%%%%%%%%%%%%%%%%%%%%%%%%%%%%%%%%%%%%%%%%%%%%
\section{Computational Complexity and Emergence}  \label{complexity}

Any Boolean function $f$ can be encoded as the ground state of an Ising Hamiltonian $H_f$. Suppose now that $f$ is intractable, then the Strong Church-Turing thesis would necessarily imply that computing a ground state of $H_f$ would also be intractable. In fact, the assertion is stronger. Since we can potentially encode the output of $f$ in the state of any spin state, the process of determining the ground state of any particular spin would also be intractable. In this final section, we will use the above intuition to provide lower bounds on the computational difficulty of the ground state problem, i.e., finding the ground state of some suitable two-dimensional, nearest neighbor Ising Hamiltonian.

In computational complexity \cite{Papadimitriou94a}, NP denotes the class of problems whose solutions can be verified, but not necessarily found, in polynomial time. It encapsulates many computational tasks that we would like to be able to solve efficiently, such as prime factoring and the traveling salesman problem \cite{Lawler87}. The hardest of such problems lie in the class NP-complete. Should any NP-complete problem be solved efficiently, then it could be used as a subroutine to efficiently solve all problems in NP and imply that $\mathrm{P} = \mathrm{NP}$. While, this remains one of the biggest theoretical questions in computer science, popular opinion tends to favor that P is distinct from NP, and hence efficient solutions of NP-complete are unlikely.

One particular well known NP-complete problem is the circuit satisfiability (CSAT) problem \cite{cook71a}: given a circuit with $n$ input bits and a single output bit described by a binary function $f$, is there a set of inputs such that the output is $1$? Consider a given CSAT problem with a circuit $f$. Theorem \ref{thm:circuit_gstate} implies that we can construct a Hamiltonian $H_f$ together with a predefined output bit $b_o$ such that $b_o = f(x)$ for any ground state of $H_f$. Since we can modify the any Hamiltonian by a constant without affecting its set of ground states, we can always choose $H_f$ such that its ground state energy is $0$.

Consider the ground state decision problem, does there exist a state with energy at most $0$ under the Hamiltonian $H'_f = H_f + (1 - b_{out})$? The perturbation $1 - b_{out}$ lifts the degeneracy in $H_f$ such that the resulting Hamiltonian $H'_f$ will have a zero energy state iff there is a set of inputs to $f$ such that it outputs $1$. Therefore, knowledge of the ground state of $H_f$ and hence $b_o$ clearly allows one to solve CSAT. Therefore, the ground state decision problem is at least NP-hard. Furthermore, since $H_f$ is a Hamiltonian on a square Ising lattice that grows at most polynomially with the size of the circuit, it is easy to check whether the energy of a given state is greater than $0$. Thus \emph{the ground state decision problem is NP-complete}.

We see that the above result, originally derived by Barahona \cite{barahona82a}, flows as a natural consequence of applying Ising lattices to solve a particular NP-complete problem. It is stimulating to speculate then, what other important results could be obtained by applying the Ising model to other non-trivial computational problems. The Halting problem \cite{Turing36a} is an exciting candidate; it and its generalizations \cite{Rice53} prove that there exist many properties of Turing machines that are undecidable. Such properties would necessarily correspond to certain properties of the Ising model, and it would be interesting to see if these properties are physically relevant.

Another promising avenue of research is to consider what the limitations on the computation of ground states imply about the macroscopic properties of the resulting Ising lattice. For example, it is easy to see how our results can be extended to show that computing the correlation length of such Ising lattices is also NP-complete. This leads to the concept of emergence in the infinite case in \cite{Gu09a}, following the path established by P. Anderson in 1972 with his celebrated paper `More is Different' \cite{Anderson72a}, where he postulated that the ground state of a spin glass may be non-computable.

Emergent properties of a physical system are properties which arise from the whole and are not deducible from the physical interactions of the component parts. In `More Really Is Different' \cite{Gu09a}, a special case of this technique was applied to show that certain macroscopic properties of a properly chosen, 2-dimensional, infinite periodic Ising lattice are emergent. That is, it is possible to embed universal circuits within infinite periodic Ising lattices, such that should certain macroscopic properties be computed, one would be able to decided whether a arbitrary computer program would halt. The result naturally motivated the question: ``What would happen should such lattices be finite?''. In this paper we see that in such sceneries these emergent macroscopic properties are connected with the known NP-complete properties of finite lattice Ising spin glasses. This relation (infinite $\rightarrow$ undecidable, finite $\rightarrow$ NP-complete) was previously proved as well in planar tiling problems \cite{Garey79a}, what suggests that it could be a common feature of complex universal systems.

%%%%%%%%%%%%%%%%%%%%%%%%%%%%%%%%%%%%%%%%%%%%%%%%%%%%%%%%%%%%%%%%%%%%%%%%%%%%%%%%%%%%%%%%%%%%%%%%%%%%%%%%%%%%%%%%%%%%%%%%%%%%%%%%%%
%%%%%%%%%%%%%%%%%%%%%%%%%%%%%%%%%%%%%%%%%%%%%%%%%%%%%%%%%%%%%%%%%%%%%%%%%%%%%%%%%%%%%%%%%%%%%%%%%%%%%%%%%%%%%%%%%%%%%%%%%%%%%
\section{Conclusion}  \label{conclusion}

We have derived the general conditions for a desired set of states to be the ground state of a classical Hamiltonian constrained to interact with a finite number of spins\textemdash including $2$-body interactions, i.e., the Ising Model. We have presented a new and simple way of encoding universal circuit computation in the ground states of Ising lattices through the construction of Ising blocks that implement the necessary logical gates and connections. This result can be immediately applied to derive a simple version of Barahona's original proof \cite{barahona82a} that the problem of finding states on Ising Hamiltonians is, in general, NP-complete.\\

%%%%%%%%%%%%%%%%%%%%%%%%%%%%%%%%%%%%%%%%%%%%%%%%%%%%%%%%%%%%%%%%%%%%%%%%%%%%%%%%%%%%%%%%%%%%%%%%%%%%%%%%%%%%%%%%%%%%%%%%%%%%%%%%%%
%%%%%%%%%%%%%%%%%%%%%%%%%%%%%%%%%%%%%%%%%%%%%%%%%%%%%%%%%%%%%%%%%%%%%%%%%%%%%%%%%%%%%%%%%%%%%%%%%%%%%%%%%%%%%%%%%%%%%%%%%%%%%
% Acknowledgments
We thank Michael Nielsen for suggesting us the topic of this work and acknowledge fruitful discussions with Christian Weedbrook, Jacob Biamonte and Chip Neville.

%%%%%%%%%%%%%%%%%%%%%%%%%%%%%%%%%%%%%%%%%%%%%%%%%%%%%%%%%%%%%%%%%%%%%%%%%%%%%%%%%%%%%%%%%%%%%%%%%%%%%%%%%%%%%%%%%%%%%%%%%%%%%%%%%%
%%%%%%%%%%%%%%%%%%%%%%%%%%%%%%%%%%%%%%%%%%%%%%%%%%%%%%%%%%%%%%%%%%%%%%%%%%%%%%%%%%%%%%%%%%%%%%%%%%%%%%%%%%%%%%%%%%%%%%%%%%%%%

\bibliographystyle{apsrev4-1}
%\bibliography{mybib}

%Merlin.mbs v4.21 2009-07-09.
%

%%%%%%%%%%%%%%%%%%%%%%%%%%%%%%%%%%%%%%%%%%%%%%%%%%%%%%%%%%%%%%%%%%%%%%%%%%%%%%%%%%%%%%%%%%%%%%%%%%%%%%%%%%%%%%%%%%%%%%%%%%%%%%%%%%
%%%%%%%%%%%%%%%%%%%%%%%%%%%%%%%%%%%%%%%%%%%%%%%%%%%%%%%%%%%%%%%%%%%%%%%%%%%%%%%%%%%%%%%%%%%%%%%%%%%%%%%%%%%%%%%%%%%%%%%%%%%%%
\begin{appendix}

\section{Derivation of equation (\ref{eqn:arel})} We first define
\begin{align}
g^\mathbf{b}_c(a_1,a_2,\ldots,a_k)&= \sum_{p=1}^k a_p \left[ \sum_{\mathbf{d} \in \mathrm{Dsc}(b,p)} H(\mathbf{d}) \right]
\end{align}

Thus for any $\mathbf{b}$ such that $\|\mathbf{b}\| > m$, we have the relation:
\begin{align}\label{eqn:hb}
H(\mathbf{b}) &= g^{\mathbf{b}}_c\left(1,1,\ldots,\alpha_m = 1, 0, \ldots,\alpha_k = 0\right) \nonumber \\
&= \sum_{k=m+1}^{\|\mathbf{b}\|-1} \left( \sum_{\mathbf{d} \in \mathrm{Dsc}(\mathbf{b},k)}c_{\mathbf{d}} \right) \nonumber \\
& + g^{\mathbf{b}}_c\left(1,1,\ldots,\alpha_m = 1, 0, \ldots,\alpha_k = 0\right)
\end{align}

Now, we note the fact
\begin{equation}
\sum_{\mathbf{d} \in \mathrm{Dsc}(b,m)} H(\mathbf{d}) =
g^{\mathbf{b}}_c\left(\beta_1,\beta_2,\ldots,\beta_{m-1},0\ldots,0\right)
\end{equation}

To compute $\beta_j$, Consider $H(\mathbf{d})$ which has exactly $^mC_j$ terms of the form $c_\mathbf{d'}$ with $\|\mathbf{d'}\| = m$. Also there exists $^{\|\mathbf{b}\|}C_m$ terms of the form $H(\mathbf{d})$. Thus to total number $c_\mathbf{d'}$ terms is $^{\|\mathbf{b}\|}C_m ^mC_j$. Dividing this by the total number of
$c_\mathbf{d'}$ that are descendant from $\mathbf{b}$ gives:
\begin{align}
\beta_j =
\frac{\ncr{\|\mathbf{b}\|}{m}\ncr{m}{j}}{\ncr{\|\mathbf{b}\|}{j}} = \ncr{\|\mathbf{b}\|-j}{\|\mathbf{b}\|-m} \qquad 1\leq j \leq m
\end{align}

So that
\begin{align}
\sum_{\mathbf{d} \in \mathrm{Dsc}(b,m)} H(\mathbf{d}) & =  g^{\mathbf{b}}_c \bigg[ \ncr{k-1}{k-m},\ncr{k-2}{k-m},\ldots \nonumber  \\
&  \ldots ,k-m+1,1,0\ldots,0 \bigg]
\end{align}

Substituting into Eq.\ (\ref{eqn:hb})
\begin{align}
H(\mathbf{b})= & \sum_{\mathbf{d} \in \mathrm{Dsc}(b,m)}
H(\mathbf{d}) + \nonumber \\
& g^{\mathbf{b}}_c\big[1-^{k-1}C_{k-m},1-^{k-2}C_{k-m},\ldots \nonumber \\
& \ldots,1-^{k-(m-2)}C_{k-m},-(k-m),0,\ldots,0\big],
\end{align}

we eliminate the $a_m$ term in the argument of $g^\mathbf{b}_c$. By writing:
\begin{align}
\sum_{\mathbf{d} \in \mathrm{Dsc}(b,m-1)} H(\mathbf{d}) = g^{\mathbf{b}}_c \bigg[ \ncr{k-1}{k-(m-1)},  \nonumber \\
\ncr{k-2}{k-(m-1)}, \ldots ,k-m+2,1,0\ldots,0 \bigg]
\end{align}

etc, we can eliminate each of $a_j$, $1\leq j \leq m$ recursively, and write out an equation for $H(d)$ entirely from the sum of its descendants:
\begin{equation}
H(\mathbf{b}) = g^\mathbf{b}_c(a_1,a_2,\ldots,a_j,\ldots,a_m = 1,0,\ldots,0)
\end{equation}

with
\begin{align}
a_m &= 1\\
a_{m-j} &= 1 - a_m\ncr{k-(m-j)}{k-m} - \nonumber \\
& a_{m-1}\ncr{k-(m-j)}{k-1} -  a_{m-2} \ncr{k-(m-j)}{k-2} -\nonumber \\
& \ldots - a_{m-(j-1)}(k-m-j)
\end{align}

substituting indices $p = m-j$, we get:
\begin{align}
a_p =& 1 - a_{p+1}(k-p) - a_{p+2}\ncr{k-p}{2} - \ldots \nonumber \\
& \ldots - a_{m}\ncr{k-p}{m-p} \qquad 1 \leq p < m
\end{align}

which is the recurrence relation featured. Thus, if $H$ is $m$-body, then the required equation is implied. Conversely, if Eq.\ (\ref{eqn:hb}) is satisfied, we have:
\begin{equation}
\sum_{k=m+1}^{\|\mathbf{b}\|-1} \left( \sum_{\mathbf{d} \in \mathrm{Dsc}(\mathbf{b},k)}c_{\mathbf{d}} \right) = 0 \qquad \forall \mathbf{b}:\, \|\mathbf{b}\| > m
\end{equation}

which has no non-trivial solutions.

\end{appendix}

\end{document}